\documentclass{article}

\usepackage{amssymb}

\usepackage{amsmath}
\usepackage{bm}
\usepackage{amsfonts}
\usepackage{amsmath}
\usepackage{mathrsfs}
\usepackage{multirow}
\numberwithin{equation}{section}
\usepackage{tablefootnote}
\usepackage{booktabs}
\usepackage{multirow}
\usepackage{siunitx}
\usepackage[flushleft]{threeparttable}
\usepackage{color}

\begin{document}
	\newtheorem{definition}{Definition}
	\newtheorem{theorem}{Theorem}
	\newtheorem{note}{Note}
	\newtheorem{example}{Example}
	\newtheorem{corollary}{Corollary}
	\newtheorem{lemma}{Lemma}
	\newtheorem{proposition}{Proposition}
	\newenvironment{proof}{{\bf Proof:\ \ }}{\qed}
	\newcommand{\qed}{\rule{0.5em}{1.5ex}}
	\newcommand{\bfg}[1]{\mbox{\boldmath $#1$\unboldmath}}

	\title{Statistical Inference for distributions with one Poisson conditional}         
	\author{Barry C. Arnold \footnote{Email: barnold@ucr.edu, University of California,  Riverside, CA 92521}    and B.G. Manjunath \footnote{Email: bgmanjunath@gmail.com, School of Mathematics and Statistics, University of Hyderabad, Hyderabad, India} }        
	\date{ }
	\maketitle
	\begin{abstract}
		It will be recalled that the classical bivariate normal distributions have normal marginals and normal conditionals. 
		It is natural to ask whether a similar phenomenon can be encountered involving Poisson marginals and conditionals. Reference to Arnold, Castillo and Sarabia's (1999) book on conditionally specified models will confirm that Poisson marginals will be encountered, together with both conditionals  being of the Poisson form, only in the case in which the variables are independent. Instead, in the present article we will be focusing on bivariate distributions with one marginal and the other family of conditionals being of the Poisson form. Such distributions are called Pseudo-Poisson distributions. We discuss distributional features of such models, explore inferential aspects and  include an examples of applications  of the Pseudo-Poisson model  to sets of over-dispersed data.	
		\vspace{0.5cm}
	\end{abstract}
	

	\section{Triangular transformation models}   
	
	We begin by reviewing a family of models called triangular transformation models which were introduced in Filus, Filus and Arnold \cite{ffa09} as follows:
	
	Let $\mathscr{F}=\{F(x;\boldsymbol{\theta}):\boldsymbol{\theta}=(\theta_1,...,\theta_m)^T \in \Theta \subset \mathscr{R}^{ m}\}$ be an $m$-parameter family of univariate distributions.  A $k$-dimensional Pseudo-$\mathscr{F}$ distribution can be constructed as follows
	$$
	P(X_{1} \leq x_{1}) = F(x_{1};\boldsymbol{\theta}_{1}) \eqno(1.1)
	$$
	and for $\ell = 2,3,\ldots,k$
	$$
	P(X_{\ell} \leq x_{\ell}|\boldsymbol{X}_{(\ell-1)} = \boldsymbol{x}_{(\ell-1)}) = F(x_{\ell};\boldsymbol{\theta}_{\ell}(\boldsymbol{x}_{(\ell-1)})) \eqno(1.2)
	$$
	where $\boldsymbol{\theta}_{1} \in \Theta$ and, for each $\ell$, \ $\boldsymbol{\theta}_{\ell}:\mathscr{R}^{ \ell-1} \rightarrow \Theta$. Note that we use the notational convention $\boldsymbol{a}_{(j)}=(a_1,a_2,...,a_j)$.
	
	\section{$k$-dimensional Pseudo-Poisson models}
	
	Within the general triangular transformation class of $k$-dimensional models can be found the class of $k$-dimensional Pseudo-Poisson models, which are more simply described in terms of discrete mass functions rather than distribution functions.
	
	Using standard notation we will write $X \sim \mathscr{P}(\lambda)$ (Poisson distribution) if, for $x\in \{0,1,2,...\}$ we have
	$$P(X=x)=\frac{e^{-\lambda} \lambda^x}{x!}.$$
	
	\begin{definition}
		A $k$-dimensional random variable $\underline{X}=(X_1,X_2,...,X_k)$ is said to have a $k$-dimensional Pseudo-Poisson distribution if there exists a positive constant $\lambda_1$ such that
		$$X_1\sim \mathscr{P}(\lambda_1)$$
		and $k-1$ functions $ \{\lambda_{\ell}:\ell=2,3,...,k\}$ where, for each $\ell$, $\lambda_{\ell}: \{0,1,2,...\}^{(\ell -1)} \rightarrow (0,\infty)$ such that
		$$X_{\ell}|\boldsymbol{X}_{(\ell-1)} = \boldsymbol{x}_{(\ell-1)} \sim \mathscr{P}(\lambda_{\ell}(\boldsymbol{x}_{(\ell-1)})).$$
	\end{definition}
	
	Note that there are no constraints on the forms of the functions $\lambda_{\ell}, \ \ \ell=2,3,...,k$ that appear in the definition, save for measurability.
	
	\bigskip
	
	In applications, it would typically be the case that the $\lambda_{\ell}$'s would be chosen to be relatively simple functions depending on a limited number of parameters.
	
	\section{Bivariate-Pseudo-Poisson models}  
	We will consider in some detail the Pseudo-Poisson models in the case in which the dimension $k$ is $2$. In that case we can describe the joint distribution as follows.
	\begin{definition}
		A $2$-dimensional random variable $\boldsymbol{X}=(X_1,X_2)$ is said to have a bivariate Pseudo-Poisson distribution if there exists a positive constant $\lambda_1$ such that
		$$X_1\sim \mathscr{P}(\lambda_1)$$
		and a functions $ \lambda_2: \{0,1,2,...\} \rightarrow (0, \infty)$ such that
		$$X_2|X_1 = x_1 \sim \mathscr{P}(\lambda_2(x_1)).$$
	\end{definition}
	
	In this case also, it might be desirable to restrict the form of the function $\lambda_2(x_1)$. For example we might restrict it to be a polynomial with unknown coefficients.
	
	\section{Some related models}  
	
	It will be recalled that the classical bivariate normal distributions have normal marginals AND normal conditionals. Thus it has $X_1\sim \mathscr{N}$ (normal distribution), $X_2 \sim \mathscr{N}$ and for each $x_1 \in \mathscr{R}$,  $X_2|X_1=x_1 \sim\mathscr{N}$, while for each $x_2\in \mathscr{R}$,  $X_1|X_2=x_2 \sim \mathscr{N}$
	also. It is natural to ask whether a similar phenomenon can be encountered with Poisson marginals and conditionals. A trivial example of this kind is one in which $X_1$ and $X_2$ are independent Poisson variables. In fact, no other examples exist. In an early paper, Seshadri and Patil \cite{sp64} argued that one could not have a non-trivial distribution with $X_1$ having a Poisson distribution and, for each $x_2$, $X_1|X_2=x_2$ also Poisson distributed. Arnold, Castillo and Sarabia \cite{acb99} (for example) consider the case in which, 
	for each $x_1 $,  $X_2|X_1=x_1 \sim\mathscr{P}$, while for each $x_2$,  $X_1|X_2=x_2 \sim \mathscr{P}$. Such distributions are called Poisson-conditionals distributions. They only have Poisson marginals in the case of independence (which could be deduced using the Seshadri-Patil result). However, if we are satisfied with having one marginal($X_1$) and the ``other" family of conditionals ($X_2|X_1=x_1$) being of the Poisson form, then the Pseudo-Poisson distributions fill the bill precisely.
	
	\section{The bivariate Pseudo-Poisson model with a linear regression function }  
	
	In this section we will consider in some detail a particularly simple bivariate Pseudo-Poisson model. For it we assume that
	
	\begin{equation} \label{pseudo-P-lin-regress}
	X_1 \sim \mathscr{P}(\lambda_1) 
	\end{equation}
	and 
	\begin{equation}
	X_2|X_1=x_1 \sim \mathscr{P}(\lambda_2+\lambda_3 x_1).
	\end{equation}
	
	The natural parameter space for this model is $\{(\lambda_1,\lambda_2,\lambda_3):\lambda_1>0,\lambda_2 > 0,\lambda_3\geq 0\}$. The case in which the variables are independent, corresponds to choice $\lambda_3=0$. Note that in the limiting case,  $\lambda_2= 0$ is a plausible value for the above model.  However, in such a framework $\lambda_3 > 0$, i.e., independence of variables is forsaken. Subsequently, when $\lambda_3=0$ then $\lambda_2$ is necessarily takes value greater than zero. With this framework the plausible parameter space for the model is $\{(\lambda_1,\lambda_2,\lambda_3):\lambda_1>0,\lambda_2 \geq 0,\lambda_3\geq 0\}$.
	
	\vspace{0.5cm}
	In the following, we derive the joint probability generating function (p.g.f.) and marginal p.g.f. of the above bivariate Pseudo-Poisson distribution. 
	\bigskip
	\begin{theorem} \label{theorem1}
		The p.g.f. for the bivariate Pseudo-Poisson distribution is given by 
		\begin{eqnarray}
		G(t_1,t_2) =  e^{\lambda_2 (t_2-1)} e^{ \lambda_1[t_1  e^{\lambda_3(t_2-1)}-1]};  \mbox{      }  t_1,t_2 \in \mathscr{R}. 
		\end{eqnarray}
	\end{theorem}
	\begin{proof}
		Given that $(X_1,X_2)$ has the following bivariate mass function 
		\begin{eqnarray}
		P(X_1=x_1,X_2=x_2) = \frac{e^{-\lambda_1}\lambda_1^{x_1}}{x_1!} \frac{e^{-(\lambda_2 + \lambda_3 x_1)} (\lambda_2 + \lambda_3 x_1)^{x_2}}{x_2!},
		\end{eqnarray}
		where $x_1,x_2 \in \{0,1,2,...\}$ and $\lambda_1>0$, $\lambda_2  \geq 0$ and $\lambda_3 \geq 0$, the bivariate p.g.f is given by 
		\begin{eqnarray*}
			G(t_1,t_2) &=& E \Big(t_1^{X_1} t_2^{X_2} \Big) \\
			&=& \sum_{x_1=0}^{\infty} \sum_{x_2=0}^{\infty} t_1^{x_1} t_2^{x_2} \frac{e^{-\lambda_1}\lambda_1^{x_1}}{x_1!} \frac{e^{-(\lambda_2 + \lambda_3 x_1)} (\lambda_2 + \lambda_3 x_1)^{x_2}}{x_2!} \\
			&=& \sum_{x_1=0}^{\infty} \Big \{ t_1^{x_1} \frac{e^{-\lambda_1}\lambda_1^{x_1}}{x_1!} \sum_{x_2=0}^{\infty} t_2^{x_2} \frac{e^{-(\lambda_2 + \lambda_3 x_1)} (\lambda_2 + \lambda_3 x_1)^{x_2}}{x_2!} \Big \} \\
			&=&  \sum_{x_1=0}^{\infty} \Big \{ t_1^{x_1} \frac{e^{-\lambda_1}\lambda_1^{x_1}}{x_1!} 
			e^{(\lambda_2 + \lambda_3 x_1)(t_2-1) } \Big \} \\
			&=& e^{-\lambda_1} e^{\lambda_2 (t_2-1)}  \sum_{x_1=0}^{\infty} \frac{\Big[t_1 \lambda_1 e^{\lambda_3(t_2-1)} \Big]^{x_1}}{x_1!} \\
			G(t_1,t_2) &=&  e^{\lambda_2 (t_2-1)} e^{ \lambda_1[t_1  e^{\lambda_3(t_2-1)}-1]}.
		\end{eqnarray*}
	\end{proof}
	
	\begin{corollary}
		The marginal mass function of $X_2$ is that of a Neyman Type A distribution, when $\lambda_2=0$.
	\end{corollary}
	
	\begin{proof}
		From Theorem \ref{theorem1} the marginal p.g.f of $X_2$ is given by 	
		\begin{eqnarray*}
			G(1,t_2) = G_{X_2}(t_2) 
			= e^{\lambda_2 (t_2-1)} e^{\lambda_1 [ e^{\lambda_3(t_2-1)}-1 ]}; \mbox{     } t_2 \in \mathscr{R} .
		\end{eqnarray*}
		When $\lambda_2=0$ the above p.g.f is given by 
		\begin{eqnarray}
		G_{X_2}(t_2)= M(t_2) = e^{\lambda_1 \Big[ e^{\lambda_3(t_2-1)}-1\Big]}.
		\end{eqnarray} 
		
	\end{proof} 
	
	\vspace{0.5cm}
	The above p.g.f. is of the form of a Neyman Type A distribution with $\lambda_3$ being the index of  clumping (see page 403 of Johnson, Kemp and Kotz \cite{jkk05}) which is also known as a Poisson mixture of Poissons distribution.

	\vspace{0.5cm}
	Now, the marginal mass function of $X_2$ is given by
	\begin{eqnarray}
	p_2(x_2)=P(X_2=x_2) = \frac{e^{-\lambda_1}\lambda_3^{x_2}}{x_2!} \sum_{j=0}^{\infty} \frac{(\lambda_1 e^{-\lambda_3})^j j^{x_2}}{j!};  \mbox{              } x_2=0,1,2,... .
	\end{eqnarray}
	i.e. $X_2$ has a Poisson distribution with the parameter $\lambda_3 \phi$ while $\phi$ itself is a random variable with the Poisson distribution with the parameter $\lambda_1$.
	
	\vspace{0.5cm}
	In the following section, we derive moments of the bivariate Pseudo-Poisson distribution. 
	
	\subsection{Moments}
	
	Now note that 
	
	\begin{eqnarray}
	E(X_1) &=& \lambda_1 \\
	E(X_2|X_1 =x_1)&=& \lambda_2 + \lambda_3 x_1 \nonumber \\	
	E\{ E(X_2|X_1 ) \} &=& \lambda_2 + \lambda_3 E(X_1) \nonumber \\
	E(X_2) &=& \lambda_2 + \lambda_3 \lambda_1,
	\end{eqnarray}
	and
	\begin{eqnarray}
	Var(X_1) &=& \lambda_1 \\
	Var(X_2|X_1 =x_1)&=& \lambda_2 + \lambda_3 x_1  \nonumber	\\
	Var(X_2) &=& E\{Var(X_2|X_1)\} + Var\{E(X_2|X_1)\}  \nonumber \\
	&=& E(\lambda_2 + \lambda_3 X_1) + Var(\lambda_2 + \lambda_3 X_1) \nonumber \\
	Var(X_2) &=& \lambda_2 + \lambda_3 \lambda_1 + \lambda^2_3 \lambda_1.
	\end{eqnarray}
	Also 
	\begin{eqnarray*}
		E(X_1 X_2) &=& E\{E(X_1 X_2|X_1)\} \\
		&=& E\{X_1E(X_2|X_1)\} \\
		&=& E(\lambda_2 X_1 + \lambda_3 X^2_1) \\
		E(X_1 X_2)	&=& \lambda_2 \lambda_1 + \lambda_3(\lambda_1 + \lambda^2_1).
	\end{eqnarray*}
	The covariance between $X_1$ and $X_2$ is thus
	\begin{eqnarray}
	Cov(X_1,X_2) &=& E(X_1X_2)- E(X_1)E(X_2) \nonumber \\
	&=& \lambda_2 \lambda_1 + \lambda_3(\lambda_1 + \lambda^2_1)  
	-  \lambda_1 ( \lambda_2 + \lambda_1 \lambda_3)	\nonumber \\
	Cov(X_1,X_2)	&=& \lambda_1 \lambda_3,
	\end{eqnarray}
	and the corresponding correlation is
	\begin{eqnarray}
	\rho &=& \frac{Cov(X_1,X_2)}{\sqrt{Var(X_1) Var(X_2)}} \nonumber \\
	&=& \frac{\lambda_1 \lambda_3}{\sqrt{\lambda_1 (\lambda_2 + \lambda_3 \lambda_1 + \lambda^2_3 \lambda_1) }}. 
	\end{eqnarray}
	
	\vspace{0.5cm}
	In the following we note three special cases which merit consideration.
	\begin{description}
		\item[Case I:]  When $\lambda_3=0$, it follows that  $\rho=0$. In fact, in this case, $X_1$ and $X_2$ are independent.
		\item[Case II:] When $\lambda_2=\lambda_3$, the correlation is of the form
		\begin{equation}
		\rho= \sqrt{\frac{\lambda_1\lambda_3}{1+\lambda_1+\lambda_1\lambda_3}}.
		\end{equation}
		\item[Case III:] In the limiting case in which $\lambda_2= 0$, $\rho$ simplifies to become
		\begin{eqnarray}
		\rho = \sqrt{\frac{\lambda_3}{1+\lambda_3}}.	
		\end{eqnarray}
		Note that this correlation only depends on $\lambda_3$, and not on $\lambda_1$, when $\lambda_2=0$. 
	\end{description}
	
	\vspace{0.5cm}
	
	We also remark that, for the case $\lambda_2=0$ , the bivariate Pseudo-Poisson distribution reduces to the bivariate Poisson-Poisson distribution.  The Poisson-Poisson distribution was originally introduced by Leiter and Hamdan \cite{lh73} in 1973 in analyzing traffic accidents and fatalities data. However, the two approaches leading to the same distribution are different.  Nevertheless, one can consider the bivariate Pseudo-Poisson as a generalization of the Poisson-Poisson distribution.

	\vspace{0.5cm}
	In the following, we state and prove a characterization of the bivariate Pseudo-Poisson distribution submodel (i.e. with $\lambda_2=0$) or the Poisson-Poisson distribution. A similar characterization for Power Series distributions is in Kyriakoussis and Papageorgiou \cite{kp89}.

	\subsection{Characterizing the Poisson-Poisson distribution}
	
	\begin{theorem}
		If $E(X_2|X_1=x_1)= \lambda_3 x_1$ and the marginal distribution of $X_2$ is of the form in $(5.6)$ (Neyman Type A Distribution) then the joint distribution of $(X_1,X_2)$ is necessarily that of the bivariate Poisson-Poisson distribution.
	\end{theorem}
	
	\begin{proof}
		Given that
		\begin{eqnarray*}
			\mu(x_1,\lambda_3):=E(X_2|X_1=x_1) = \sum_{x_2=0}^{\infty} x_2 f_2(x_2|x_1) = \sum_{x_2=0}^{\infty} x_2\frac{f_1(x_1|x_2) p_2(x_2)}{p_1(x_1)}.
		\end{eqnarray*}
		Therefore,
		\begin{eqnarray}
		\mu(x_1,\lambda_3) p_1(x_1) &=& \sum_{x_2=0}^{\infty} x_2 f_1(x_1|x_2) p_2(x_2) \nonumber \\
		\mu(x_1,\lambda_3) \sum_{x_2=0}^{\infty} f_1(x_1|x_2) p_2(x_2) &=& \sum_{x_2=0}^{\infty} x_2 f_1(x_1|x_2) p_2(x_2).
		\end{eqnarray}
		Now, using $p_2(x_2)$ in the above equation we have
		\begin{eqnarray*}
			\mu(x_1,\lambda_3) \sum_{x_2=0}^{\infty} \Bigg \{ f_1(x_1|x_2) \frac{e^{-\lambda_1}\lambda_3^{x_2}}{x_2!} \sum_{j=0}^{\infty} \frac{(\lambda_1 e^{-\lambda_3})^j j^{x_2}}{j!} \Bigg\} =  \\ \sum_{x_2=0}^{\infty} \Bigg \{ x_2 f_1(x_1|x_2) \frac{e^{-\lambda_1}\lambda_3^{x_2}}{x_2!} \sum_{j=0}^{\infty} \frac{(\lambda_1 e^{-\lambda_3})^j j^{x_2}}{j!} \Bigg\}.
		\end{eqnarray*}
		Now,
		\begin{eqnarray}
		\mu(x_1,\lambda_3) \sum_{x_2=0}^{\infty} f_1(x_1|x_2) e^{-\lambda_1} \lambda_3^{x_2} \nu(x_2) = \sum_{x_2=0}^{\infty}x_2 f_1(x_1|x_2) e^{-\lambda_1} \lambda_3^{x_2} \nu(x_2)
		\end{eqnarray}
		where 
		\begin{eqnarray*}
			\nu(x_2) = \frac{1}{x_2!} \sum_{j=0}^{\infty} \frac{(\lambda_1 e^{-\lambda_3})^j j^{x_2}}{j!}.
		\end{eqnarray*}
		Denote
		\begin{eqnarray*}
			h(x_1,\lambda_3) = \sum_{x_2=0}^{\infty} f_1(x_1|x_2) \lambda_3^{x_2} \nu(x_2).
		\end{eqnarray*}
		Therefore, $(5.16)$ becomes
		\begin{eqnarray}
		\mu(x_1,\lambda_3) h(x_1,\lambda_3) &=& \lambda_3 h'(x_1,\lambda_3) \nonumber \\
		\frac{h'(x_1,\lambda_3)}{h(x_1,\lambda_3)} &=& \frac{\mu(x_1,\lambda_3)}{\lambda_3}
		\end{eqnarray}
		where $h'(x_1,\lambda_3) = \frac{\partial}{\partial \lambda_3} h(x_1,\lambda_3)$. Now, the solution to the above differential equation is 
		
		\begin{eqnarray*}
			h(x_1,\lambda_3) &=& c(x_1) \exp \Bigg \{ \int  \frac{\mu(x_1,\lambda_3)}{\lambda_3} d\lambda_3 \Bigg\}	 \\
			&=& c(x_1)  \exp \Bigg \{ \int  \frac{\lambda_3 x_1}{\lambda_3} d\lambda_3 \Bigg\} \\
			&=& c(x_1) e^{\lambda_3 x_1}.
		\end{eqnarray*}	
		
		Now,
		\begin{eqnarray}
		h(x_1,\lambda_3) = \sum_{x_2=0}^{\infty} f_1(x_1|x_2) \lambda_3^{x_2} \nu(x_2) &=& c(x_1) e^{\lambda_3 x_1} \nonumber \\
		\sum_{x_2=0}^{\infty} f_1(x_1|x_2) \lambda_3^{x_2} \nu(x_2) &=& c(x_1) \sum_{x_2=0}^{\infty}  \frac{(\lambda_3 x_1)^{x_2}}{x_2!}.
		\end{eqnarray}
		From the above equation comparing the coefficients of $\lambda_3^{x_2}$, we have
		\begin{eqnarray*}
			f_1(x_1|x_2) \nu(x_2) &=& c(x_1) \frac{x_1^{x_2}}{x_2!} \\
			f_1(x_1|x_2) &=& \frac{c(x_1)}{\nu(x_2)} \frac{x_1^{x_2}}{x_2!}.
		\end{eqnarray*}
		Now, substituting for $\nu(x_2)$ we get
		\begin{eqnarray}
		f_1(x_1|x_2) = \frac{c(x_1) x_1^{x_2}}{\sum_{j=0}^{\infty} \frac{(\lambda_1 e^{-\lambda_3})^j j^{x_2}}{j!}}; \mbox{        } x_1=0,1,2,... .	
		\end{eqnarray}
		To find the value of $c(x_1)$, by summing over $x_1$, we get
		\begin{eqnarray*}
			\sum_{x_1=0}^{\infty} f_1(x_1|x_2)= 1 &=& \frac{\sum_{x_1=0}^{\infty} c(x_1)x_1^{x_2}}{\sum_{j=0}^{\infty} \frac{(\lambda_1 e^{-\lambda_3})^j}{j!}j^{x_2}}; \mbox{        } \forall x_2=0,1,2,...\nonumber \\
			\sum_{j=0}^{\infty} \frac{(\lambda_1 e^{-\lambda_3})^j}{j!} j^{x_2} &=& \sum_{x_1=0}^{\infty} c(x_1) x_1^{x_2}; \mbox{        } \forall x_2=0,1,2,...
		\end{eqnarray*}
		Now,  the above equality can also be rewritten as
		\begin{eqnarray}
		\sum_{i=0}^{\infty} \Bigg[ c(i) - \frac{(\lambda_1 e^{-\lambda_3})^i}{i!}  \Bigg] i^{x_2} =0; \mbox{        } \forall x_2=0,1,2,... .
		\end{eqnarray}
		The above equality will be satisfied if 
		\begin{eqnarray}
		c(i)= \frac{(\lambda_1 e^{-\lambda_3})^i}{i!}, \mbox{       } \forall i.
		\end{eqnarray}
		Therefore, $(5.19)$ will be
		\begin{eqnarray}
		f_1(x_1|x_2) = \frac{(\lambda_1 e^{-\lambda_3})^{x_1}}{x_1!} \frac{x_1^{x_2}}{\sum_{j=0}^{\infty} \frac{(\lambda_1 e^{-\lambda_3})^j j^{x_2}}{j!}},
		\end{eqnarray}
		so that, 
		\begin{eqnarray*}
			f(x_1,x_2) = P(X_1=x_1,X_2=x_2) = f_1(x_1|x_2) p_2(x_2) = \frac{e^{-\lambda_1}\lambda_1^{x_1}}{x_1!} \frac{e^{-\lambda_3 x_1}(\lambda_3 x_1)^{x_2}}{x_2!},
		\end{eqnarray*}
		as claimed.
	\end{proof}

	\subsection{Fisher Dispersion Index}
	In this section, we derive the Fisher dispersion index for the bivariate Pseudo-Poisson distribution.  In the present note, we use the definition of the bivariate Fisher dispersion index provided by Kokonendji and Puig \cite{kp18}.
	\vspace{0.5cm}
	The marginal dispersion indices are
	\begin{eqnarray}
	DI(X_1) &=& \frac{Var(X_1)}{E(X_1)} = 1  \mbox{  (equi-dispersion)}. \\
	DI(X_2) &=& \frac{Var(X_2)}{E(X_2)} = \frac{\lambda_2 + \lambda_3 \lambda_1 + \lambda_3^2 \lambda_1}{\lambda_2 + \lambda_3 \lambda_1} \nonumber \\ &=& 1 +  \frac{\lambda_3^2 \lambda_1}{\lambda_2 + \lambda_3 \lambda_1} \mbox{  ( over-dispersion)  }.
	\end{eqnarray}
	
	We state and prove the following theorem.

	\begin{theorem}
		The bivariate Pseudo-Poisson distribution is always over-dispersed.
	\end{theorem}
	\begin{proof}
		Define, for the Pseudo-Poisson model
		\begin{eqnarray}
		E(\textbf{X}) = (\lambda_1, \lambda_2 + \lambda_3 \lambda_1)^T
		\end{eqnarray}
		\[
		cov(\textbf{X})=
		\begin{bmatrix}
		\lambda_1 &  \lambda_1 \lambda_3 \\
		\lambda_1 \lambda_3 & \lambda_2 + \lambda_3 \lambda_1 + \lambda^2_3 \lambda_1
		\end{bmatrix}.
		\]

		Now,
		
		\begin{eqnarray*}
			E(\textbf{X})^TE(\textbf{X}) &=& \lambda_1^2 + (\lambda_2 + \lambda_3 \lambda_1)^2 \\
			\sqrt{E(\textbf{X})}^T (cov(\textbf{X})) \sqrt{E(\textbf{X})}) &=& \lambda_1^2 + 2 \lambda_1^{\frac{3}{2}} \lambda_3 \sqrt{\lambda_2 + \lambda_3 \lambda_1} + \\ && (\lambda_2 +  \lambda_3\lambda_1) (\lambda_2 + \lambda_3 \lambda_1 + \lambda_3^2 \lambda_1).
		\end{eqnarray*}
		
		Using the definition given in Kokonendji and Puig \cite{kp18} page 183, we have.
		
		\begin{eqnarray}
		GDI(\textbf{X}) &=& \frac{\lambda_1^2 + 2 \lambda_1^{\frac{3}{2}} \lambda_3 \sqrt{\lambda_2 + \lambda_3 \lambda_1} + (\lambda_2 +  \lambda_3\lambda_1) (\lambda_2 + \lambda_3 \lambda_1 + \lambda_3^2 \lambda_1)}{\lambda_1^2 + (\lambda_2 + \lambda_3 \lambda_1)^2} \nonumber \\
		&=& 1 + \frac{ 2 \lambda_1^{\frac{3}{2}} \lambda_3 \sqrt{\lambda_2 + \lambda_3 \lambda_1} + (\lambda_2 +  \lambda_3\lambda_1)  \lambda_3^2 \lambda_1}{\lambda_1^2 + (\lambda_2 + \lambda_3 \lambda_1)^2}>1,
		\end{eqnarray}
		which indicates over-dispersion as claimed.
	\end{proof}
	
	\vspace{0.5cm}
	
	Note that in a set of bivariate count data, if one marginal is equi-dispersed and other is over-dispersed, one can consider the bivariate Pseudo-Poisson distribution as a possible model.

	\vspace{0.5cm}

	\bigskip
	
	\bigskip
	\section{Statistical Inference}
	In this section  we obtain moment and maximum likelihood estimators of parameters  $\lambda_1$,$\lambda_2$ and $\lambda_3$. In addition, we construct the likelihood ratio test for the simpler submodel, i.e., for $\lambda_2=\lambda_3$. Finally, we consider a simulation study and a real-life application of the bivariate Pseudo-Poisson distribution.
	\subsection{Moments and Moment Estimators}
	Now suppose that we have data of the form $\boldsymbol{X}^{(1)},\boldsymbol{X}^{(2)},...,\boldsymbol{X}^{(n)}$ which are i.i.d. with common distribution  (5.1)-(5.2). Note that, for each $i$, $\boldsymbol{X}^{(i)}=(X_{1i},X_{2i})^T$. Method of moments estimators of the parameters are readily derived. 
	The respective simpler submodels and their statistical inference are also considered in the following section.
	\vspace{.5cm}

	Now, if we equate the sample means and the sample covariance to their expectations, and if $M_1>0$, we obtain the following
	consistent asymptotically normal method of moments estimates.

	\begin{eqnarray}
	\tilde{\lambda}_1 &=& M_1 \\
	\tilde{\lambda}_2 &=& M_2-S_{12} \\
	\tilde{\lambda}_3 &=& \frac{S_{12}}{M_1}
	\end{eqnarray}
	where
	\begin{eqnarray*}
		M_1 &=& \frac{1}{n}\sum_{i=1}^n X_{1i}	\\
		M_2 &=& \frac{1}{n}\sum_{i=1}^n X_{2i},
	\end{eqnarray*}
	and
	\begin{eqnarray*}
		S_{12}=\frac{1}{n} \sum_{i=1}^n (X_{1i}-M_1)(X_{2i}-M_2).
	\end{eqnarray*}
	
	\bigskip
	
	If we consider the simpler sub-model in which $\lambda_2=0,$ then the method of moments estimates of the remaining two $\lambda$'s are even simpler. Thus, again provided that $M_1>0$,
	
	\begin{eqnarray}
	\tilde{\lambda}_1 &=& M_1 \\
	\tilde{\lambda}_3 &=& \frac{M_2}{M_1}.
	\end{eqnarray} 
	
	For the sub-model in which $\lambda_2=\lambda_3$, method of moments estimates of the parameters are given by
	\begin{eqnarray}
	\tilde{\lambda}_1 &=& M_1 \\
	\tilde{\lambda}_3 &=& \frac{M_2}{1+M_1}.
	\end{eqnarray}

	\subsection{Maximum Likelihood Estimators}
	
	In the two-parameter model (i.e. when $\lambda_2=0$ or $\lambda_2=\lambda_3$) the maximum likelihood estimates can be verified to coincide with the method of moments estimates derived in the previous subsection, provided that $M_1>0$.  
	
	Identifying the maximum likelihood estimator's (m.l.e.) in the three parameter model is a little more challenging. For the given data of the form $\boldsymbol{X}^{(1)},\boldsymbol{X}^{(2)},...,\boldsymbol{X}^{(n)}$ which are i.i.d. with common distribution  (5.1)-(5.2) then the likelihood function is as follows
	
	\begin{eqnarray} \label{ljpp}
	L(\boldsymbol{\theta}) &=& \prod_{i=1}^{n} \Bigg \{ \frac{e^{-\lambda_1} \lambda_1^{x_{1i}}}{x_{1i}!}  
	\frac{e^{-(\lambda_2 + \lambda_3 x_{1i})} (\lambda_2 + \lambda_3 x_{1i})^{x_{2i}}}{x_{2i}!} \Bigg \} \nonumber \\
	&=& \frac{e^{-n(\lambda_1 + \lambda_2)} \lambda_1^{\sum_{i=1}^n x_{1i}} e^{-\lambda_3 \sum_{i=1}^n x_{1i}} \prod_{i=1}^{n} (\lambda_2 + \lambda_3 x_{1i})^{x_{2i}} }{\prod_{i=1}^{n}\{(x_{1i})! (x_{2i})!\}}
	\end{eqnarray}
	where $\boldsymbol{\theta}=(\lambda_1,\lambda_2,\lambda_3)^T$.
	
	The corresponding log-likelihood function is 
	\begin{eqnarray}  \label{lljpp}
	\l = \log L(\boldsymbol{\theta}) &=&  -n (\lambda_1 + \lambda_2) + \log(\lambda_1) \sum_{i=1}^n x_{1i} -\lambda_3 \sum_{i=1}^n x_{1i} \nonumber \\ &&+ \sum_{i=1}^n x_{2i} \log(\lambda_2 + \lambda_3 x_{1i}) + h(\boldsymbol{x}_1,\boldsymbol{x}_2)
	\end{eqnarray} 
	where $h(\boldsymbol{x}_1,\boldsymbol{x}_2) =\log \Big (\frac{1}{\prod_{i=1}^{n}\{(x_{1i})! (x_{2i})!\}} \Big )$.
	
	\bigskip
	Now, differentiating with respect to the $\lambda_i$'s we get the following the likelihood equations
	
	\begin{eqnarray} \label{mle}
	-n+\frac{1}{\lambda_1}\sum_{i=1}^n X_{1i} &=& 0  \label{lik-eq-1}\\
	-n+\sum_{i=1}^n\frac{X_{2i}}{\lambda_2+\lambda_3X_{1i}} &=& 0  \label{lik-eq-2}\\
	-\sum_{i=1}^nX_{1i}+\sum_{i=1}^n\frac{X_{1i}X_{2i}}{\lambda_2+\lambda_3X_{1i}} &=& 0. \label{lik-eq-3}
	\end{eqnarray}
	
	If $M_1=(1/n)\sum_{i=1}^nX_{1i}=0$, then there is no solution to (\ref{lik-eq-1}), otherwise Equation (\ref{lik-eq-1}) is readily solved, yielding the m.l.e. for $\lambda_1$, namely
	\begin{eqnarray}
	\hat{\lambda}_1=M_1.
	\end{eqnarray}
	
	The remaining two equations must be solved numerically (provided that $M_1>0$), to obtain $\hat{\lambda}_2$ and $\hat{\lambda}_3$.

	\bigskip
	
	\subsection{Likelihood Ratio Test}
	As usual, the general form of a generalized likelihood ratio test statistic is of the form
	\begin{eqnarray} \label{lrt}
	\Lambda = \frac{\sup_{ \boldsymbol{\theta} \in \Theta_0} L(\boldsymbol{\theta})}{\sup_{\boldsymbol {\theta} \in \Theta} L(\boldsymbol{\theta})}.
	\end{eqnarray}
	
	Here, $\Theta_0$ is a subset of $\Theta$ and we envision testing $H_0:\boldsymbol{\theta} \in \Theta_0$. We reject the null hypothesis for a small value of $\Lambda$.
	
	\vspace{0.5cm}
	In the following section we construct likelihood ratio test for the simpler submodels.

	\subsubsection{Submodel I: For $\lambda_2=\lambda_3$, equivalently, testing for $H_0: \lambda_2=\lambda_3$}
	
	The natural parameter space under the full model is
	$\Theta =\{(\lambda_1,\lambda_2,\lambda_3)^T: \lambda_1>0, \lambda_2 \geq 0,\lambda_3 \geq 0 \}$. 
	Besides, under the null hypothesis the natural parameter space is  $\Theta_0 = \{(\lambda_1,\lambda_3)^T: \lambda_1>0,\lambda_3 > 0 \}$.
	
	\vspace{0.5cm}
	
	Under $H_0$, equation (\ref{lljpp}) will be
	
	\begin{equation} 
	l =  -n (\lambda_1 + \lambda_3) + \log(\lambda_1) \sum_{i=1}^n x_{1i} 
	-\lambda_3 \sum_{i=1}^n x_{1i} + \sum_{i=1}^n x_{2i} \log[ \lambda_3 (x_{1i}+1)] + h(\boldsymbol{x}_1,\boldsymbol{x}_2).
	\end{equation}

	Now,  taking partial derivatives with respect to $\lambda_1$ and $\lambda_3$ and equating them  to zero, we get the following equations
	
	\begin{eqnarray*} 
		-n + \frac{1}{\lambda_1} \sum_{i=1}^n X_{1i} &=& 0 \\
		-n-\sum_{i=1}^n X_{1i} + \sum_{i=1}^n \frac{X_{2i}}{ \lambda_3} &=& 0.
	\end{eqnarray*}
	
	Then, the m.l.e's of $\lambda_1$ and $\lambda_3$ are, provided that $M_1>0$, given by 
	
	\begin{eqnarray*}
		\widehat{\lambda}^*_1 &=& M_1 \\
		\widehat{\lambda}^*_3 &=& \frac{M_2}{1+M_1}.
	\end{eqnarray*}
	
	\noindent Note that these estimates agree with the method of moments estimates given in Section 6.1.
	
	\vspace{0.5cm}
	
	Now, in the unrestricted parameter space $\Theta$, i.e., under the full model, the m.l.e's for $\lambda_1$, $\lambda_2$ and $\lambda_3$ are obtained from equations (\ref{lik-eq-1})--(\ref{lik-eq-3}). 
	
	\vspace{0.5cm}
	
	Let 
	$\widehat{\lambda}_1$,$\widehat{\lambda}_2$ and $\widehat{\lambda}_3$ be the respective m.l.e's of $\lambda_i$'s then the generalized likelihood ratio test statistic defined in (\ref{lrt}) will be
	
	\begin{eqnarray*} 
		\Lambda_1 = \frac{\frac{e^{-n(	\widehat{\lambda}^*_1  + \widehat{\lambda}^*_3 )} (	\widehat{\lambda}^*_1)^{\sum_{i=1}^n x_{1i}} e^{-\widehat{\lambda}^*_3 \sum_{i=1}^n x_{1i}} \prod_{i=1}^{n} [\widehat{\lambda}^*_3 (1+x_{1i})]^{x_{2i}} }{\prod_{i=1}^{n}\{(x_{1i})! (x_{2i})!\}}}{\frac{e^{-n(\widehat{\lambda}_1 + \widehat{\lambda}_2)} (\widehat{\lambda}_1)^{\sum_{i=1}^n x_{1i}} e^{-\widehat{\lambda}_3 \sum_{i=1}^n x_{1i}} \prod_{i=1}^{n} (\widehat{\lambda}_2 + \widehat{\lambda}_3 x_{1i})^{x_{2i}} }{\prod_{i=1}^{n}\{(x_{1i})! (x_{2i})!\}}}.
	\end{eqnarray*}
	
	Since $\widehat{\lambda}_1 = 	\widehat{\lambda}^*_1$,  the above test statistic simplifies to become
	
	\begin{eqnarray} 
	\Lambda_1 = e^{n \widehat{\lambda}_2 } e^{-(	\widehat{\lambda}^*_3 -\widehat{\lambda}_3 ) \sum_{i=1}^n x_{1i}} \prod_{i=1}^{n} \Bigg [\frac{\widehat{\lambda}^*_3 (1+x_{1i})}{\widehat{\lambda}_2 + \widehat{\lambda}_3 x_{1i}} \Bigg ]^{x_{2i}}.
	\end{eqnarray}

	Now by taking the logarithm, we have
	\begin{eqnarray}
	\log \Lambda_1 = n \widehat{\lambda}_2 - (\widehat{\lambda}^*_3 - \widehat{\lambda}_3 ) \sum_{i=1}^n x_{1i} + \sum_{i=1}^n x_{2i} \log \Bigg [ \frac{\widehat{\lambda}^*_3 (1+x_{1i})}{\widehat{\lambda}_2 + \widehat{\lambda}_3 x_{1i}}  \Bigg ].
	\end{eqnarray}
	
	\vspace{0.5cm}
	
	If $n$ is large, then $-2\log \Lambda_1$ may be compared with a suitable $\chi^2_1$ percentile in order to decide whether $H_0$ should be accepted.
	
	\subsubsection{Submodel II: For $\lambda_2=0$, equivalently, testing for $H_0: \lambda_2=0$}
	
	Under the null hypothesis the natural parameter space is  $\Theta_0 = \{(\lambda_1,\lambda_3)^T: \lambda_1>0,\lambda_3 > 0 \}$ and the m.l.e's of $\lambda_1$ and $\lambda_3$ are,
	provided that $M_1>0$, given by 
	
	\begin{eqnarray*}
		\widehat{\lambda}^*_1 &=& M_1 \\
		\widehat{\lambda}^*_3 &=& \frac{M_2}{M_1}.
	\end{eqnarray*}
	
	Also, these estimates coincide  with the method of moments estimates given in Section 6.1.

	\vspace{0.5cm}
	
	The natural parameter space under the full model is
	$\Theta =\{(\lambda_1,\lambda_2,\lambda_3)^T: \lambda_1>0, \lambda_2 \geq 0,\lambda_3 \geq 0 \}$ and $\widehat{\lambda}_1$, $\widehat{\lambda}_2$ and $\widehat{\lambda}_3$ are respective  m.l.e's obtained from equation (\ref{lik-eq-1})--(\ref{lik-eq-3}).
	
	\vspace{0.5cm}
	
	Therefore, the generalized likelihood ratio test statistic defined in (\ref{lrt}) will be
	
	\begin{eqnarray*} 
		\Lambda_2 = \frac{\frac{e^{-n(	\widehat{\lambda}^*_1  + \widehat{\lambda}^*_3 )} (	\widehat{\lambda}^*_1)^{\sum_{i=1}^n x_{1i}} e^{-\widehat{\lambda}^*_3 \sum_{i=1}^n x_{1i}} \prod_{i=1}^{n} (\widehat{\lambda}^*_3 x_{1i})^{x_{2i}} }{\prod_{i=1}^{n}\{(x_{1i})! (x_{2i})!\}}}{\frac{e^{-n(\widehat{\lambda}_1 + \widehat{\lambda}_2)} (\widehat{\lambda}_1)^{\sum_{i=1}^n x_{1i}} e^{-\widehat{\lambda}_3 \sum_{i=1}^n x_{1i}} \prod_{i=1}^{n} (\widehat{\lambda}_2 + \widehat{\lambda}_3 x_{1i})^{x_{2i}} }{\prod_{i=1}^{n}\{(x_{1i})! (x_{2i})!\}}}.
	\end{eqnarray*}
	
	Since $\widehat{\lambda}_1 = 	\widehat{\lambda}^*_1$, then the above test statistic becomes 
	
	\begin{eqnarray} 
	\Lambda_2 = e^{n \widehat{\lambda}_2 } e^{-(	\widehat{\lambda}^*_3 -\widehat{\lambda}_3 ) \sum_{i=1}^n x_{1i}} \prod_{i=1}^{n} \Bigg (\frac{\widehat{\lambda}^*_3 x_{1i}}{\widehat{\lambda}_2 + \widehat{\lambda}_3 x_{1i}} \Bigg )^{x_{2i}}.
	\end{eqnarray}

	Taking logarithm, we have
	\begin{eqnarray}
	\log \Lambda_2 = n \widehat{\lambda}_2 - (\widehat{\lambda}^*_3 - \widehat{\lambda}_3 ) \sum_{i=1}^n x_{1i} + \sum_{i=1}^n x_{2i} \log \Bigg (\frac{\widehat{\lambda}^*_3 x_{1i}}{\widehat{\lambda}_2 + \widehat{\lambda}_3 x_{1i}}  \Bigg ).
	\end{eqnarray}
	
	\vspace{0.5cm}
	
	If $n$ is large, then $-2\log \Lambda_2$ may be compared with a suitable $\chi^2_1$ percentile in order to decide whether $H_0$ should be accepted.

	\subsubsection{ Testing for independence, i.e.,  $H_0: \lambda_3=0$}
	
	Under the null hypthesis the natural parameter space is  $\Theta_0 = \{(\lambda_1,\lambda_2)^T: \lambda_1>0,\lambda_2 > 0 \}$ and the m.l.e's of $\lambda_1$ and $\lambda_2$ are,
	provided that $M_1>0$,  given by 
	
	\begin{eqnarray*}
		\widehat{\lambda}^*_1 &=& M_1 \\
		\widehat{\lambda}^*_2 &=& M_2.
	\end{eqnarray*}
	
	Also, these estimates be coincide with the method of moments estimates as given in Section 6.1.

	\vspace{0.5cm}
	
	The natural parameter space under the full model is
	$\Theta =\{(\lambda_1,\lambda_2,\lambda_3)^T: \lambda_1>0, \lambda_2  \geq 0,\lambda_3 \geq 0 \}$ and $\widehat{\lambda}_1$, $\widehat{\lambda}_2$ and $\widehat{\lambda}_3$ are the respective  m.l.e's obtained from equation (\ref{lik-eq-1})--(\ref{lik-eq-3}).
	
	\vspace{0.5cm}
	
	Therefore, the generalized likelihood ratio test statistic defined in (\ref{lrt}) will be
	
	\begin{eqnarray*} 
		\Lambda_3 = \frac{\frac{e^{-n(	\widehat{\lambda}^*_1  + \widehat{\lambda}^*_2 )} (	\widehat{\lambda}^*_1)^{\sum_{i=1}^n x_{1i}} e^{-\widehat{\lambda}^*_2 \sum_{i=1}^n x_{2i}}   }{\prod_{i=1}^{n}\{(x_{1i})! (x_{2i})!\}}}{\frac{e^{-n(\widehat{\lambda}_1 + \widehat{\lambda}_2)} (\widehat{\lambda}_1)^{\sum_{i=1}^n x_{1i}} e^{-\widehat{\lambda}_3 \sum_{i=1}^n x_{1i}} \prod_{i=1}^{n} (\widehat{\lambda}_2 + \widehat{\lambda}_3 x_{1i})^{x_{2i}} }{\prod_{i=1}^{n}\{(x_{1i})! (x_{2i})!\}}}.
	\end{eqnarray*}
	
	Since $\widehat{\lambda}_1 = 	\widehat{\lambda}^*_1$, then the above test statistic becomes 
	
	\begin{eqnarray} 
	\Lambda_3 = e^{n \widehat{\lambda}_2 } e^{-(	\widehat{\lambda}^*_2 -\widehat{\lambda}_2 ) \sum_{i=1}^n x_{1i}} \prod_{i=1}^{n} \Bigg (\frac{\widehat{\lambda}^*_2 }{\widehat{\lambda}_2 + \widehat{\lambda}_3 x_{1i}} \Bigg )^{x_{2i}}.
	\end{eqnarray}

	Taking logarithm, we have
	\begin{eqnarray}
	\log \Lambda_3 = n \widehat{\lambda}_2 - (\widehat{\lambda}^*_2 - \widehat{\lambda}_2 ) \sum_{i=1}^n x_{1i} + \sum_{i=1}^n x_{2i} \log \Bigg (\frac{\widehat{\lambda}^*_2 }{\widehat{\lambda}_2 + \widehat{\lambda}_3 x_{1i}}  \Bigg ).
	\end{eqnarray}
	
	\vspace{0.5cm}
	
	If $n$ is large then $-2\log \Lambda_3$ may be compared with a suitable $\chi^2_1$ percentile in order to decide whether $H_0$ should be accepted.

	\subsection{Examples}
	In the following two sub-sections we provide a simulation study and give examples  of real-life applications of the bivariate Pseudo-Poisson distribution.
	
	\subsubsection{Simulation data}
	Simulating from Pseudo models is straightforward because of the marginal and conditional structure of the model. In the following we give a simple simulation algorithm for the bivariate Pseudo-Poisson model with linear regression. For a given $\lambda_1$, $\lambda_2$ and $\lambda_3$. 
	\begin{description}
		\item[Step 1:] Simulate $x_1$ from $\mathscr{P}(\lambda_1)$.
		\item[Step 2:] Simulate $x_2$ from $\mathscr{P}(\lambda_2 + \lambda_3 x_1)$.
	\end{description}
	Repeat the above two steps for the desired number of observations.
	\bigskip
	
	We have simulated  $10,000$ data sets of sample size $n= 50,100,500,1000$  for the parameter values $\lambda_1=1$, $\lambda_2=3$ and $\lambda_3 =4$.  The corresponding moment and m.l.e's and also their bootstrapped standard errors are displayed in the  \\
	Table  \ref{simulation_data} \footnote{SE: Standard Error;     PC:Pearson Correlation}.  Note that with increase in sample size the moment and m.l.e.'s standard error (SE) decreases and the Pearson correlation (PC) converges to the population correlation.

	\begin{table}
		\caption{Simulation}  
		\label{simulation_data}
		\small 
		\centering 
		\begin{tabular}{lcccccr} 
			\toprule[\heavyrulewidth]\toprule[\heavyrulewidth]
			\textbf{$n$ }  & \textbf{Parameter} & \textbf{Moment} & \textbf{MLE} & \textbf{SE(Moment)} & \textbf{SE(MLE)} & \textbf{PC} \\ 
			\midrule
			\multirow{5}{*}{$50$} & $\lambda_1$ & $1.000$ & $1.000$ & $0.142$ & $0.142$ & \multirow{5}{*}{$0.831$}\\
			& $\lambda_2$ & $3.086$ & $3.000$ & $0.896$ & $0.390$ \\
			& $\lambda_3$ & $3.911$ & $3.998$ & $0.898$ & $0.421$\\
			& $\rho$ & $0.814$ & $0.830$& $0.063$ & $0.027$ \\
			\bottomrule[\heavyrulewidth] 
			\multirow{5}{*}{$100$} & $\lambda_1$ & $0.999$ & $0.999$ & $0.099$ & $0.099$ & \multirow{5}{*}{$0.833$}\\
			& $\lambda_2$ & $3.048$ & $3.002$ & $0.645$ & $0.271$ \\
			& $\lambda_3$ & $3.953$ & $3.998$ & $0.651$ & $0.292$\\
			& $\rho$ & $0.824$ & $0.832$ & $0.043$ & $0.018$\\
			\bottomrule[\heavyrulewidth] 
			\multirow{5}{*}{$500$} & $\lambda_1$ & $1.000$ & $1.000$ & $0.044$& $0.044$ & \multirow{5}{*}{$0.834$}\\
			& $\lambda_2$ & $3.009$ & $3.002$ & $0.288$ & $0.121$\\
			& $\lambda_3$ & $3.990$ & $3.998$ & $0.290$ & $0.128$\\
			& $\rho$ & $0.832$ & $0.834$ & $0.018$ & $0.008$\\
			\bottomrule[\heavyrulewidth]     	
			\multirow{5}{*}{$1000$} & $\lambda_1$ & $1.000$  & $1.000$ & $0.032$ &$0.032$ & \multirow{5}{*}{$0.838$}\\
			& $\lambda_2$ & $3.002$ & $3.000$ & $0.206$ & $0.086$\\
			& $\lambda_3$ & $3.997$ & $3.999$ & $0.208$ & $0.091$\\
			& $\rho$ & $0.833$ & $0.834$ & $0.013$ & $0.005$\\
			\bottomrule[\heavyrulewidth]  	  
			
			\bigskip

		\end{tabular}
	\end{table}

	\subsubsection{A particular data set  I}
	
	We consider a data set which is mentioned in Islam and Chowdhury \cite{ic17}, the source of the data is from  the tenth wave of the Health and Retirement Study (HRS).
	The data represents the number of conditions ever had $(X_1)$ as mentioned by the doctors and utilization of healthcare services (say, hospital, nursing home, doctor and home care) $(X_2)$.
	
	It has been noted that the sample Pearson correlation coefficient for the above data is $0.063$. Primarily, for further analyses, the data has been tested for independence (c.f. Section 6.3.3) and the $-2 \log \Lambda_3$  value is $28.359$. Consequently, the assumption of variables independence is rejected. 
	
	Further, the estimated Fisher index of $X_1$ is $0.801$ (approximately equi-dispersed) and the dispersion index of $X_2$ is $1.03$ (slightly over-dispersed).  Moment and m.l.e's values are displayed in Table \ref{HRS_model1}.  Next,we  consider the sub-models  and fit the same data for these models. Recall that for these sub-models the
	m.l.e's and the moment estimates coincide,
	
	\begin{itemize}
		\item {\bf Sub-Model I:}  That is, when $\lambda_2=\lambda_3$, for fitted values, c.f.  Table \ref{HRS_model2}.
		\item {\bf Sub-Model II:} For $\lambda_2=0$,  the fitted values are displaced in Table \ref{HRS_model3}.
	\end{itemize}

	\begin{table}
		\caption{Health and Retirement Study Data: Full Model}  
		\label{HRS_model1}
		\small 
		\centering 
		\begin{tabular}{lccccr} 
			\toprule[\heavyrulewidth]\toprule[\heavyrulewidth]
			\textbf{n }  & \textbf{Parameter} & \textbf{Moment}  & \textbf{m.l.e} &  \textbf{PC} &  \textbf{-2 log L} \\ 
			\midrule
			\multirow{5}{*}{$5567$} & $\lambda_1$ & $2.643$ & $2.643$ & \multirow{5}{*}{$0.063$} & \multirow{5}{*}{$32766.08$}\\
			& $\lambda_2$ & $0.688$ & $0.64$   \\
			& $\lambda_3$ & $0.031$ & $0.049$ \\
			& $\rho$ & $ 0.057$ & $0.091$ \\
			\bottomrule[\heavyrulewidth] 
		\end{tabular}
		
	\end{table}
	
	\medskip

	\begin{table}
		\caption{Health and Retirement Study Data: Sub-Model I}  
		\label{HRS_model2}
		\small 
		\centering 
		\begin{tabular}{lccr} 
			\toprule[\heavyrulewidth]\toprule[\heavyrulewidth]
			\textbf{n }  & \textbf{Parameter} &  \textbf{Moment}   &   \textbf{-2 log L} \\ 
			\midrule
			\multirow{3}{*}{$5567$} & $\lambda_1$ & $2.643$  & \multirow{3}{*}{$33077.09$}\\
			& $\lambda_3$ & $0.211$  \\
			
			\bottomrule[\heavyrulewidth] 
		\end{tabular}
		\bigskip
		\medskip

		\caption{Health and Retirement Study Data: Sub-Model II}  
		\label{HRS_model3}
		\small 
		\centering 
		\begin{tabular}{lccr} 
			\toprule[\heavyrulewidth]\toprule[\heavyrulewidth]
			\textbf{n }  & \textbf{Parameter} &  \textbf{Moment}   &   \textbf{-2 log L} \\ 
			\midrule
			\multirow{3}{*}{$5567$} & $\lambda_1$ & $2.643$  & \multirow{3}{*}{$43813.17$}\\
			& $\lambda_3$ & $0.291$  \\
			\bottomrule[\heavyrulewidth] 
		\end{tabular}
	\end{table}
	
	Note that using the $AIC$ criteria  Pseudo-Poisson Full-Model fit the data better. 
	
	\subsubsection{A particular data set  II}
	
	Here, we consider a data set which is in Leiter and Hamdan \cite{lh73}, the source of the data is a 50-mile stretch of Interstate 95 in Prince William, Stafford and Spottsylvania counties in easter Virginia. The data represents the number of accident categorized as fatal accidents, injury accidents or property damage accidents, along with the corresponding number of fatalities and injuries for the period 1 January 1969 to 31 October 1970.
	
	We consider the number of fatalities as $X_1$, since the estimated Fisher index is $1.051$ and the number of injury accidents as $X_2$ (estimated Fisher index is $1.141$).  Moment and m.l.e's values are displayed in Table \ref{shuterdata_model1}.  Next,we  consider the sub-models  and fit the same data for these models. Recall that for this sub-models the
	m.l.e's and the moment estimates coincide,
	
	\begin{itemize}
		\item {\bf Sub-Model I:}  That is, when $\lambda_2=\lambda_3$, for fitted values, see  Table \ref{shuterdata_model2}.
		\item {\bf Mirrored  Sub-Model II (c.f. Section 7):} For $\lambda_2=0$,  the fitted values are displaced in Table \ref{shuterdata_model3}. 
	\end{itemize}

	\bigskip
	Note that using the AIC criteria the Pseudo-Poisson Mirrored Sub-Model II fits the data better.

	\begin{table}
		\caption{Accidents and Fatalities Data: Full Model}  
		\label{shuterdata_model1}
		\small 
		\centering 
		\begin{tabular}{lccccr} 
			\toprule[\heavyrulewidth]\toprule[\heavyrulewidth]
			\textbf{n }  & \textbf{Parameter} & \textbf{Moment}  & \textbf{m.l.e} &  \textbf{PC} &  \textbf{-2 log L} \\ 
			\midrule
			\multirow{5}{*}{$639$} & $\lambda_1$ & $0.058$ & $0.058$ & \multirow{5}{*}{$0.205$} & \multirow{5}{*}{$1862.076$}\\
			& $\lambda_2$ & $0.812$ & $0.813$   \\
			& $\lambda_3$ & $0.867$ & $0.843$ \\
			& $\rho$ & $0.219$ & $0.213$ \\
			\bottomrule[\heavyrulewidth] 
		\end{tabular}
		
	\end{table}
	\medskip

	\begin{table}
		\caption{Accidents and Fatalities Data: Sub-Model I}  
		\label{shuterdata_model2}
		\small 
		\centering 
		\begin{tabular}{lccr} 
			\toprule[\heavyrulewidth]\toprule[\heavyrulewidth]
			\textbf{n }  & \textbf{Parameter} &  \textbf{Moment}   &   \textbf{-2 log L} \\ 
			\midrule
			\multirow{3}{*}{$639$} & $\lambda_1$ & $0.058$  & \multirow{3}{*}{$1862.094$}\\
			& $\lambda_3$ & $0.815$  \\
			\bottomrule[\heavyrulewidth] 
		\end{tabular}
		
		\bigskip
		\medskip
		
		\caption{Accidents and Fatalities Data: Mirrored Sub-Model II}  
		\label{shuterdata_model3}
		\small 
		\centering 
		\begin{tabular}{lccr} 
			\toprule[\heavyrulewidth]\toprule[\heavyrulewidth]
			\textbf{n }  & \textbf{Parameter} &  \textbf{Moment}   &   \textbf{-2 log L} \\ 
			\midrule
			\multirow{3}{*}{$639$} & $\lambda_1$ & $0.862$  & \multirow{3}{*}{$1847.505$}\\
			& $\lambda_3$ & $0.067$  \\
			
			\bottomrule[\heavyrulewidth] 
			
		\end{tabular}
	\end{table}

	\section{The mirrored, or permuted model}
	
	When we assume that $X_1\sim \mathscr{P}(\lambda_1)$ and that $X_2|X_1=x_1 \sim\mathscr{P}(\lambda_2(x_1))$, it is natural to think that, in some unspecified way, the variable $X_1$ influences or, dare we say, causes $X_2$. But, for many data sets the ordering of the variables is quite arbitrary and we should also entertain the possibility that the data might be better modeled by the corresponding "mirrored" model in which 
	$X_2\sim \mathscr{P}(\lambda_1)$ and $X_1|X_2=x_2 \sim\mathscr{P}(\lambda_2(x_2))$. The original model and the mirrored model are distinct and, inevitably, one of them will fit the data better than the other (except in the less interesting case in which $X_1$ and $X_2$ are independent).
	
	\bigskip

	With this in mind, we return to the Health and Retirement Study data. The corresponding values
	of the AIC criterion for our original Pseudo-Poisson, Sub-models are displayed in Table \ref{aic}, together with those for the corresponding mirrored models and the Bivariate Conway-Maxwell-Poisson (COM- Poisson) model.  We refer to Sellers et. al.
	\cite{kdb16} for the further discussion on the bivariate COM-Poisson model. Using the Akaike information criterion, the   Bivariate COM-Poisson model appears to be the best  but the computation time required for fitting this model may be a problem.  Also, note that BPP MSM-II is not suitable for the Health and Retirement Study data since the Pseudo-Poisson model is only appropriate when $X_1=0$ (or mirrored  $X_2=0$) implies that $X_2=0$ (or mirrored  $X_1=0$).
	\bigskip
	
	Now, for the Accidents and Fatalities data, note that  the considered data is not suitable for the Mirrored Full model or the Sub-Model II.  We refer to Table  \ref{aic2} for AIC values of other models.  Using the Akaike information criterion, the mirrored Bivariate Pseudo-Poisson Sub-Model II 
	model appears to be the best.  For the Accidents and Fatalities data the models BPP MFM and  BPP SM-II are inappropriate. Also, note that the bivariate  Pseudo-Poisson mirrored Sub-model II is exactly the same model as that considered in Leiter and Hamdan \cite{lh73}.

	\bigskip
	There do exist other over-dispersed models which  include the bivariate COM-Poisson as a special case. However, the number of parameters to fit the data and the computation time for the analysis are less in the Pseudo-Poisson model. For example, the bivariate COM-Poisson model has $6$ parameters and for the above given data size computation is very slow  because of the non-existence of closed-form expressions.    Also, note that for the Pseudo-Poisson Sub-Model I, both $X_1$ and $X_2$ can take any non-negative integer values but such data sets are not plausible for the Sub-Model II or the Poisson-Poisson model in Leiter and Hamdan \cite{lh73} or its mirrored models. Finally, we reiterate our recommendation that the bivariate Pseudo-Poisson model should be used when the given count data has one marginal equi-dispersed and the other over-dispersed.

	\begin{table}
		\caption{Health and Retirement Study Data: AIC}  
		\label{aic}
		\small 
		\centering 
		
		\begin{tabular}{lcr} 
			\toprule[\heavyrulewidth]\toprule[\heavyrulewidth]
			\textbf{Models }  & \textbf{No. Parameters}  & \textbf{AIC}  \\ 
			\midrule
			BPP FM & $3$ & $32772.08$ \\
			BPP MFM & $3$ & $32783.08$ \\
			BPP SM-I  & $2$ & $33081.09$  \\
			BPP MSM-I  & $2$ &  $35640.46$\\
			BPP SM-II & $2$ & $43817.17$ \\
			BPP MSM-II & $2$ &  $----$\\
			BCMP & 6 & $32690.18$\\
			\bottomrule[\heavyrulewidth] 
			
			\bigskip
			
		\end{tabular}
		
		\begin{tablenotes}
			\item {\bf AIC values for Bivariate Pseudo-Poisson Full Model (BPP FM),  Bivariate Pseudo-Poisson Mirrored Full Model (BPP MFM) ,Bivariate Pseudo-Poisson Sub-Model I  (BPP SM-I), Bivariate Pseudo-Poisson Mirrored Sub-Model I (BPP  MSM-I),  Bivariate Pseudo-Poisson Sub-Model II (BPP SM-II), Bivariate Pseudo-Poisson Mirrored Sub-Model II (BPP  MSM-II) and Bivariate COM-Poisson (BCMP) on Health and Retirement Study data.}
		\end{tablenotes}
	\end{table}

	\bigskip

	\begin{table}
		\caption{Accidents and Fatalities Data: AIC}  
		\label{aic2}
		\small 
		\centering 
		
		\begin{tabular}{lcr} 
			\toprule[\heavyrulewidth]\toprule[\heavyrulewidth]
			\textbf{Models }  & \textbf{No. Parameters}  & \textbf{AIC}  \\ 
			\midrule
			BPP FM & $3$ & $1862.076$ \\
			BPP MFM &$3$ & $----$ \\
			BPP SM-I  & $2$ & $1866.094$  \\
			BPP MSM-I  & $2$ &  $1865.560$\\
			BPP SM-II & $2$ & $----$ \\
			BPP MSM -II & $2$ &  $1847.505$\\
			BCMP & 6 & $1854.125$\\
			\bottomrule[\heavyrulewidth] 
			
			\bigskip
			
		\end{tabular}
	\end{table}

	\bigskip
	
	\begin{note} Having analyzed a bivariate data set using a Pseudo-Poisson model, the analysis of the corresponding mirrored model can be implemented by repeating the analysis with the roles of the $X_{1i}$'s and $X_{2i}$'s interchanged.
	\end{note}	
	
	\section{Permutations of $k$-variate models}
	
	In the discussion of $k$-variate pseudo-models in Sections $1$ and $2$, in reality it will be better to consider $k!$ related models obtained by permuting the roles of the $k$ variables in the data set to be fitted. In such a situation it will usually be the case that one and only one of the $k!$ models will turn out to provide the best fit to the data.

\end{document}